\documentclass[submission,copyright,creativecommons]{eptcs}
\providecommand{\event}{AFL 2017}
\usepackage{underscore}
\usepackage{graphicx}
\usepackage{amsmath}
\usepackage{subfig}
\usepackage{enumitem}
\usepackage{pifont}
\usepackage{amsthm}
\usepackage{array}
\newcolumntype{P}[1]{>{\centering\arraybackslash}p{#1}}
\newcommand{\xmark}{\ding{55}}%
\newtheorem{theorem}{Theorem}
\newtheorem{proposition}{Proposition}
\newtheorem{lemma}{Lemma}
\newtheorem{remark}{Remark}

\title{A New Sensing $5'\rightarrow 3'$  Watson-Crick Automata Concept}
\author{Benedek Nagy
\institute{Department of Mathematics, Faculty of Arts and Sciences, \\
      Eastern Mediterranean University, Famagusta, North Cyprus, Mersin-10, Turkey\\}
\email{nbenedek.inf@gmail.com}
\and
Shaghayegh Parchami
\institute{Department of Mathematics, Faculty of Arts and Sciences, \\
      Eastern Mediterranean University, Famagusta, North Cyprus, Mersin-10, Turkey\\}
\email{\quad shaghayegh@gmail.com}
\and
Hamid Mir-Mohammad-Sadeghi
\institute{Department of Mathematics, Faculty of Arts and Sciences, \\
      Eastern Mediterranean University, Famagusta, North Cyprus, Mersin-10, Turkey\\}
\email{\quad ha.sadeghi@gmail.com}
}

\def\titlerunning{A New Sensing $5'\rightarrow 3'$  Watson-Crick Automata Concept}
\def\authorrunning{B. Nagy, S. Parchami \& H. Mir-Mohammad-Sadeghi}

\begin{document}
\maketitle

\begin{abstract}
Watson-Crick (WK) finite automata are working on a Watson-Crick tape, that is, on a DNA molecule. Therefore, it has two reading heads. While in traditional WK automata both heads read the whole input in the same physical direction, in $5'\rightarrow 3'$  WK automata the heads start from the two extremes and read the input in opposite direction. In sensing  $5'\rightarrow 3'$  WK automata the process on the input is finished when the heads meet. Since the heads of a WK automaton may read longer strings in a transition, in previous models a so-called sensing parameter took care for the proper meeting of the heads (not allowing to read the same positions of the input in the last step). In this paper, a new model is investigated, which works without the sensing parameter (it is done by an appropriate change of the concept of configuration). Consequently, the accepted language classes of the variants are also changed.  Various hierarchy results are proven in the paper.

\end{abstract}

\section{Introduction}
\label{s:intro}
DNA computing provides relatively new paradigms of computation \cite{Adleman,Paun} from the end of the last century. In contrast, automata theory is one of the base of computer science. Watson-Crick-automata (abbreviated as WK automata), as a branch of DNA computing was introduced in \cite{Freund}; they relate to both mentioned fields: they have important relation to formal language and automata theory. More details can be found in \cite{Paun} and \cite{Czeizle}. WK automata work on double-stranded tapes called Watson-Crick tapes (i.e., DNA molecules), whose strands are scanned separately by read-only heads. The symbols in the corresponding cells of the double-stranded tapes are related by (the Watson-Crick) complementarity relation. The relationships between the classes of the Watson-Crick automata are investigated in \cite{Freund,Paun,Kuske}. The two strands of a DNA molecule have opposite $5'\rightarrow 3'$  orientation. Considering the reverse and the $5'\rightarrow 3'$  variants, they are more realistic in the sense, that both heads use the same biochemical direction (that is opposite physical directions) \cite{Freund,DNA2008,Leupold}. Some variations of the reverse Watson-Crick automaton with sensing power which tells whether the upper and the lower heads are within a fixed distance (or meet at the same position) are discussed in \cite{DNA2008,Nagy2009,iConcept,Nagy2013}. 
Since the heads of a WK automaton may read longer strings in a transition, in these models the sensing parameter took care of the proper meeting of the heads by sensing if the heads are close enough to meet in the next transition step.

The motivation of the new model is to erase the rather artificial term of sensing parameter from the model. By the sensing parameter one can `cheat' to allow only special finishing transitions, and thus,  in the old model the all-final variants have the same accepting power as the variants without this condition. Here, the accepted language classes of the new model are analyzed. Variations such as
all-final, simple, 1-limited, and stateless $5'\rightarrow 3'$  Watson-Crick automata are also detailed.

\section{Preliminaries, Definitions}
\label{s:pre}
We assume that the reader is familiar with basic concepts of formal languages and automata, otherwise she or he is referred to \cite{Handb}. We denote the empty word by $\lambda$.

The two strands of the DNA molecule have opposite $5'\rightarrow3'$ orientations. For this reason, it is worth to take into account a variant of Watson-Crick finite automata that parse the two strands of the Watson-Crick tape in opposite directions. Figure \ref{sd} indicates the initial configuration of such an automaton.

\begin{figure}[h]
    \centering
        \includegraphics[scale=0.29]{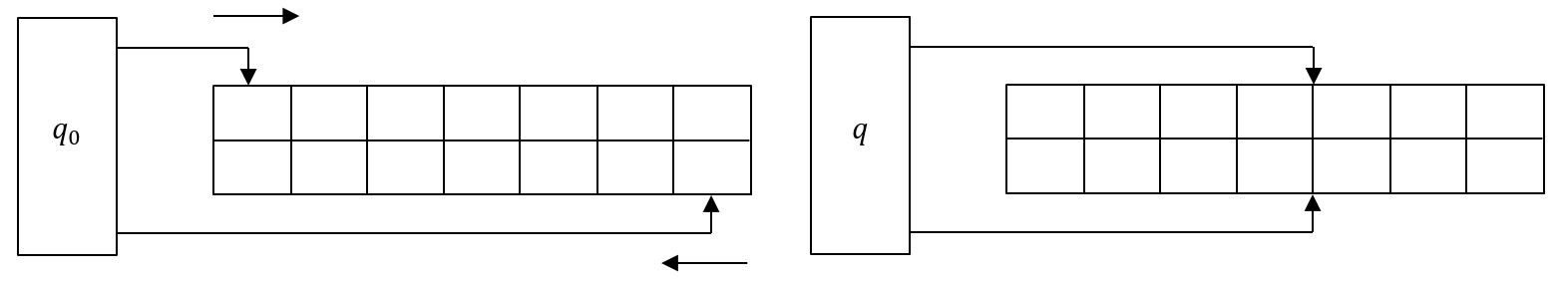}
            \caption{A sensing $5'\rightarrow 3'$ WK automaton in the initial configuration and in an accepting configuration (with a final state $q$).}
    \label{sd}
\end{figure}

The $5'\rightarrow3'$ WK automaton is sensing, if the heads sense that they are meeting. \\
Formally, a Watson-Crick automaton is a 6-tuple $M=(V,\rho,Q,q_0,F,\delta)$, where: 
\begin{itemize}
\item $V$ is the (input) alphabet, 
\item $\rho\subseteq V\times V$ denotes a complementarity relation, 
\item $Q$ represents a finite set of states, 
\item $q_0\in Q$ is the initial state, 
\item $F\subseteq Q$ is the set of final (accepting) states and 
\item $\delta$ is called transition mapping and it is of the form  $\delta: Q \times \left(\begin{array}{c}V^{*}\\ V^{*}\end{array}\right)\rightarrow 2^Q$, such that it is non empty only for finitely many triplets $(q,u,v), q \in Q, u,v\in V^*$.
    \end{itemize}
In sensing $5'\rightarrow3'$ WK automata every pair of positions in the Watson-Crick tape is read by exactly one of the heads in an accepting computation, and therefore the complementarity relation cannot play importance, instead, we assume that it is the identity relation. Thus, it is more convenient to consider the input as a normal word instead the double stranded form. Note here that complementarity can be excluded from the traditional models as well, see \cite{Kuske} for details.\\
Let us define the radius of an automaton by $r$ which shows the maximum length of the substrings of the input that can be read by the automaton in a transition.
A configuration of a Watson-Crick automaton is a pair $(q,w)$ where $q$ is the current state of the automaton and $w$ is the part of the input word which has not been processed (read) yet. For $w',x,y\in V^*,q,q'\in Q$, we write a transition between two configurations as:\\
$(q,xw'y)\Rightarrow(q',w' )$ if and only if $q'\in \delta(q,x,y)$. We denote the reflexive and transitive closure of the relation $\Rightarrow$ by $\Rightarrow^*$. Therefore, for a given $w\in V^*$, an accepting computation is a sequence of transitions $(q_0,w) \Rightarrow^* (q_F,\lambda)$, starting from the initial state and ending in a final state.
The language accepted by a WK automaton $M$ is:\\
$L(M)=\{ w\in V^*  \mid  (q_0,w) \Rightarrow^* (q_F,\lambda),  q_F\in F$\}.
The shortest nonempty word accepted by $M$ is denoted by $w_s$, if it is uniquely determined or any of them if there are more than one such word(s).\\
There are some restricted versions of WK automata which can be defined as follows: 
\begin{itemize}
\item $\textbf{N}$: stateless, i.e., with only one state: if $Q=F=\{q_0\}$;
\item $\textbf{F}$: all-final, i.e., with only final states: if $Q=F$;
\item $\textbf{S}$: simple (at most one head moves in a step) $\delta:(Q\times ((\lambda,V^* )\cup(V^*,\lambda)))\rightarrow 2^Q$.
\item $\textbf{1}$: 1-limited (exactly one letter is being read in each step) $\delta:(Q\times ((\lambda,V)\cup (V,\lambda)))\rightarrow2^Q$.
\end{itemize}
Additional versions can be determined using multiple constrains such as \textbf{F1}, \textbf{N1}, \textbf{FS}, \textbf{NS} WK automata.

Now, as an example, we show the language $L=\{a^nb^m\mid n,m\geq 0 \}$ that can be accepted by an $\textbf{N1}$ sensing $5'\rightarrow3'$ WK automaton (Figure \ref{tm:2a}).
\begin{figure}[h]
    \centering
        \includegraphics[scale=0.15]{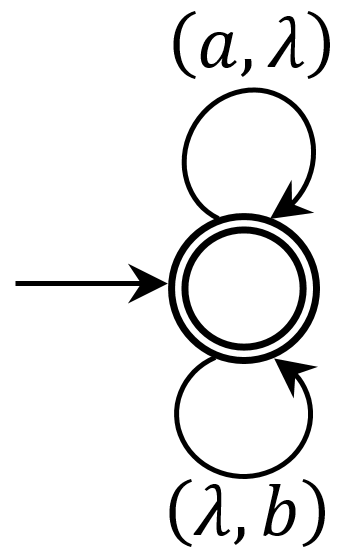}
            \caption{A sensing $5' \rightarrow 3'$ WK automaton of type $\textbf{N1}$ accepting the language $\{a^nb^m\mid n,m\geq 0 \}$.}
    \label{tm:2a}
\end{figure}

\section{Hierarchy by sensing $5'\rightarrow 3'$ WK automata}
\label{s:WK}

\begin{theorem}\label{thm:1} The following classes of languages coincide:
\begin{itemize}
\item the class of linear context-free languages defined by linear context-free grammars,
\item the language class accepted by sensing $5'\rightarrow 3'$ WK finite automata,
\item the class of languages accepted by $\textnormal{\textbf{S}}$ sensing $5'\rightarrow 3'$ WK automata,
\item the class of languages accepted by $\textnormal{\textbf{1}}$ sensing $5'\rightarrow 3'$ WK automata.
\end{itemize}
\end{theorem}

\begin{proof} For DNA computing reasons (and for simplicity) we work with $\lambda$-free languages.
The proof is constructive, first we show that the first class is included in the last one. Let $G=(N,T,S,P)$ be a linear context-free grammar having productions only in the forms $A\to aB, A\to Ba, A\to a$ with $A,B\in N,\  a\in T$. Then the $\textbf{1}$ sensing $5'\rightarrow 3'$ WK automaton $M=(T,id,N\cup\{q_f\},S,\{q_f\},\delta)$ is defined with $B\in \delta(A,u,v)$ if $A\to uBv \in P$ and
$q_f \in \delta(A,u,\lambda)$ if $A\to u \in P$ ($u,v\in T\cup\{\lambda\}$). Clearly, each (terminated) derivation in $G$ coincides to a(n accepting) computation of $M$, and vice versa. Thus the first class is included in the last one.

The inclusions between the fourth, third and second classes are obvious by definition. To close the circle, we need to show that the second class is in the first one. Let the sensing $5'\to 3'$ WK automaton $M=(V,id,Q,q_0,F,\delta)$ be given. Let us construct the linear context-free grammar
$G=(Q,V,q_0,P)$ with productions: $p\to u q v $ if $q\in\delta(q,u,v)$
and $p\to u v \in P$ if $q\in\delta(q,u,v)$ and $q\in F$
($p,q\in Q,\  u,v\in V^*$). Again, the (accepting) computations of $M$ are in a bijective correspondence to the (terminated) derivations in $G$. Thus, the proof is finished.
\end{proof}

Based on the previous theorem we may assume that the considered sensing $5'\rightarrow 3'$ WK automata have no $\lambda$-movements, i.e., at least one of the heads is moving in each transition.

\begin{lemma}\label{l1}
	Let $M$ be an $\textnormal{\textbf{F1}}$ sensing $5'\rightarrow 3'$ WK automaton and let the word $w \in V^+$ that is in $L(M)$. Let $\left|w\right|=k$, then for each $l$, where $0\leq l \leq k$, there is at least one word $w_l\in L(M)$ such that $\left|w_l\right|=l$.	
\end{lemma}

\begin{proof}
	According to the definition of $\textbf{F1}$ sensing $5'\rightarrow 3'$ WK automaton, $w$ can be accepted in $k$ steps such that in each step, the automaton can read exactly one letter. Moreover, each state is final, therefore by considering the first $l$ steps of the $k$ steps, the word $w_l=w'_lw''_l$ is accepted by $M$, where $w'_l$ is read by the left head and $w''_l$ is read by the right head during these $l$ steps, respectively.
\end{proof}

\begin{remark}\label{R1}
	Since, by definition, every $\textnormal{\textbf{N1}}$ sensing $5'\rightarrow 3'$ WK automaton is $\textnormal{\textbf{F1}}$ WK automaton at the same time, Lemma \ref{l1} applies for all $\textnormal{\textbf{N1}}$ sensing $5' \rightarrow 3'$ WK automata also.
\end{remark}

\begin{theorem}\label{thm:2}
	The class of languages that can be accepted by $\textnormal{\textbf{N1}}$ sensing $5'\rightarrow 3'$ WK automata is properly included in the language class accepted by $\textnormal{\textbf{NS}}$ sensing $5'\rightarrow 3'$ WK automata.
\end{theorem}

\begin{proof}
	Obviously, these automata have exactly one state. In $\textbf{NS}$ machines, the reading head may read some letters in a transition, while the input should be read letter by letter by $\textbf{N1}$ machines. The language $L=\{a^{3n} b^{2m}\mid n,m\geq 0\}$ proves the proper inclusion. In this language $w_s$ is $bb$ and in an $\textbf{NS}$ automaton it can be accepted by any of the following transitions: $(bb,\lambda)$, $(\lambda,bb)$. Although by Lemma \ref{l1}, $w_s$ cannot be the shortest nonempty accepted word in a language accepted by an $\textbf{N1}$ sensing $5'\rightarrow 3'$ WK automaton. Figure \ref{tm:2} shows that language $L$ can be accepted by an $\textbf{NS}$ sensing $5'\rightarrow 3'$ WK automaton. Therefore, the proper inclusion stated in the theorem is proven.
\end{proof}

\begin{figure}[h]
    \centering
        \includegraphics[scale=0.15]{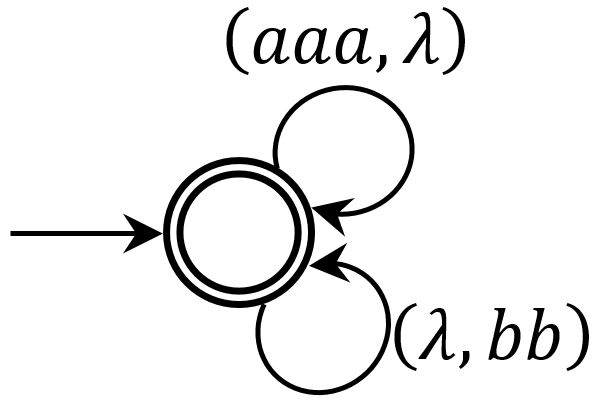}
            \caption{A sensing $5' \rightarrow 3'$ WK automaton of type $\textbf{NS}$ accepting the language $\{a^{3n} b^{2m}\mid n,m\geq 0\}$.}
    \label{tm:2}
\end{figure}

\begin{theorem}\label{thm:3}
	The class of languages that can be accepted by $\textnormal{\textbf{NS}}$ sensing $5'\rightarrow 3'$ WK automata is properly included in the language class accepted by $\textnormal{\textbf{N}}$ sensing $5'\rightarrow 3'$ WK automata.
\end{theorem}

\begin{proof}
The language $L=\{a^{(2n+m)} b^{(2m+n)}\mid n,m\geq0\}$ proves the proper inclusion. Suppose that there is an $\textbf{NS}$ sensing $5'\rightarrow 3'$ WK automaton that accepts $L$. The $\textbf{NS}$ sensing $5'\rightarrow 3'$ WK automaton has exactly one state and  one of the heads can move at a time. The $w_s$ of $L$ is $aab$ (or $abb$). It can be accepted by one of the following loop transitions: $(aab,\lambda)$, $(\lambda,aab)$, $(abb,\lambda)$ or $(\lambda,abb)$ by an $\textbf{NS}$ sensing $5'\rightarrow 3'$ WK automaton. Each of the mentioned transitions can lead to accept different language from the language $\{a^{(2n+m)} b^{(2m+n)}\mid n,m \geq 0\}$. For instance, using several times the transition $(aab,\lambda)$, the language $\{(aab)^n\mid n\geq0\}$ is accepted which is not a subset of the language $L$.
Therefore, the language $\{a^{(2n+m)} b^{(2m+n)}\mid n,m\geq0\}$ cannot be accepted by $\textbf{NS}$ sensing $5'\rightarrow 3'$ WK automata. Figure \ref{tm:3} shows that this language can be accepted by an $\textbf{N}$ sensing $5'\rightarrow 3'$ WK automaton. Hence, the theorem holds.
\end{proof}

\begin{figure}[h]
    \centering
        \includegraphics[scale=0.15]{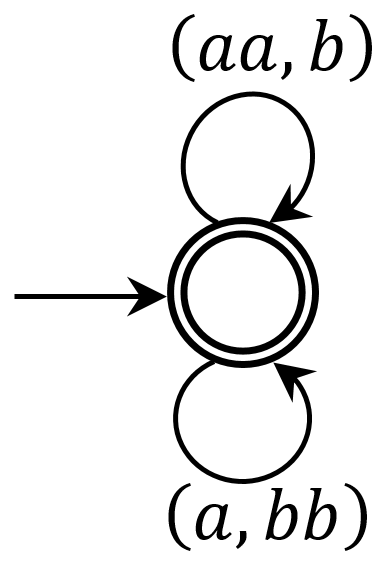}
            \caption{An $\textbf{N}$ sensing $5' \rightarrow 3'$ WK automaton of type $\textbf{N}$ accepting the language $\{a^{2n+m} b^{2m+n}\mid n,m\geq 0\}$.}
    \label{tm:3}
\end{figure}

Now the concept of sensing $5' \rightarrow 3'$ WK automata with sensing parameter is recalled \cite{Nagy2009,Nagy2013}. Formally, a 6-tuple $M=(V,\rho,Q,q_0,F,\delta')$ is a sensing $5' \rightarrow 3'$ WK automaton with sensing parameter,
where, $V$, $\rho$, $Q$, $q_0$ and $F$ are the same as in our model and $\delta'$ is the transition mapping
 defined by the sensing condition in the following way: \\
$\delta': \left(Q \times \left(\begin{array}{c}V^{*}\\ V^{*}\end{array}\right) \times D\right)\rightarrow 2^Q$, where the sensing distance set is indicated by $D=\{0,1,\dots,r,+\infty\}$ where $r$ is the radius of the automaton.
In $\delta'$, the distance between the two heads is used from the set $D$ if it is between $0$ and $r$, and $+\infty$ is used, when the distance of the two heads is more than $r$. In this way,  the set $D$ is an efficient tool and it controls the appropriate meeting of the heads: 
When the heads are close to each other only special transitions are allowed.
\\

The next three theorems highlight the difference between the new model and the model with sensing parameter.
\begin{theorem}\label{thm:4}
	The class of languages that can be accepted by $\textnormal{\textbf{F1}}$ sensing $5' \rightarrow 3'$ WK automata is properly included in the language class of $\textnormal{\textbf{FS}}$ sensing $5' \rightarrow 3'$ WK automata.
\end{theorem}

\begin{proof}
	Obviously, all states of these automata are final and $\textbf{F1}$ sensing $5' \rightarrow 3'$ WK automata should read the input letter by letter, while $\textbf{FS}$ sensing $5' \rightarrow 3'$ WK automata may read some letters in a transition. To show proper inclusion, consider the language $L=\{(aa)^n (bb)^m\mid  m \leq n\leq m+1,m\geq 0\}$. The word $w_s$ can be $aa$ and by Lemma \ref{l1}, $w_s$ cannot be the shortest nonempty accepted word for an $\textbf{F1}$ sensing $5' \rightarrow 3'$ WK automaton. However, $L$ can be accepted by an $\textbf{FS}$ sensing $5' \rightarrow 3'$ WK automaton as it is shown in Figure \ref{tm:4}. The theorem is proven.
\end{proof}
\begin{figure}[h]
    \centering
        \includegraphics[scale=0.15]{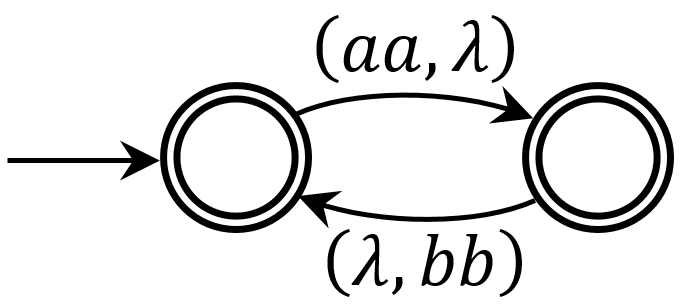}
            \caption{A sensing $5' \rightarrow 3'$ WK automaton of type $\textbf{FS}$ accepting the language $\{(aa)^n (bb)^m\mid  m \leq n\leq m+1,m\geq 0\}$.}
    \label{tm:4}
\end{figure}

\begin{theorem}\label{thm:5}
	The language class accepted by $\textnormal{\textbf{FS}}$ sensing $5' \rightarrow 3'$ WK automata is properly included in the language class of $\textnormal{\textbf{F}}$ sensing $5' \rightarrow 3'$ WK automata.
\end{theorem}

\begin{proof}
The language $L=\{a^{2n+q} c^{4m}b^{2q+n}\mid n,q\geq 0,m\in \{0,1\} \}$ proves the proper inclusion. Let us assume, contrary that $L$ is accepted by an $\textbf{FS}$ sensing $5' \rightarrow 3'$ WK automaton.
Let the radius of this automaton be $r$. Let $w=a^{2n+q}b^{2q+n}\in L$ with  $ n,q\geq r$ such that $|w|=3n+3q>r$. Then the word $w$ cannot be accepted by using only one of the transitions (from the initial state $q_0$), i.e., $\delta(q_0,a^{2n+q}b^{2q+n},\lambda)$ or $\delta(q_0,\lambda,a^{2n+q}b^{2q+n})$ is not possible.
Therefore, by considering the position of the heads after using any of the transitions from the initial state $q_0$ in $\textbf{FS}$ sensing $5' \rightarrow 3'$ WK automaton (all states are final and one of the heads can move),
 it is clear that either a prefix or a suffix of $w$ with length at most $r$ is accepted by the automaton.
But neither a word from $a^+$, nor from $b^+$ is in $L$.
This fact contradicts to our assumption, 
hence $L$ cannot be accepted by any $\textbf{FS}$ sensing $5' \rightarrow 3'$ WK automata. However, it can be accepted by $\textbf{F}$ $5' \rightarrow 3'$ WK automata, since the two heads can move at the same time and they can read both blocks of $a$'s and $b$'s simultaneously. In Figure \ref{tm:5}, an all-final $5' \rightarrow 3'$ WK automaton can be seen which accepts $L$.
\end{proof}
\begin{figure}[h]
    \centering
        \includegraphics[scale=0.15]{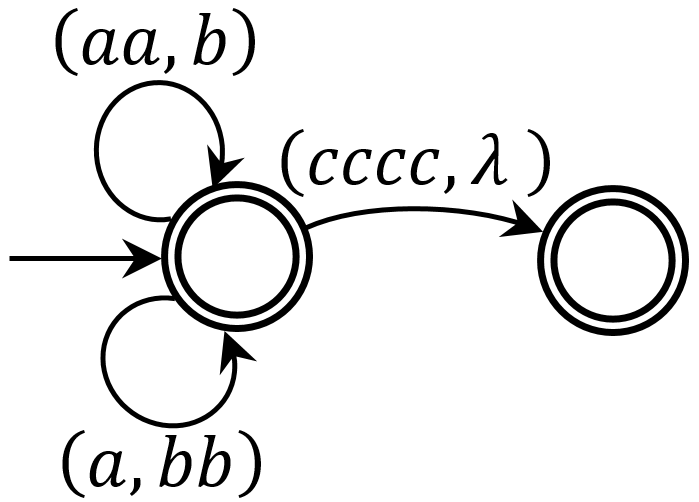}
            \caption{A sensing $5' \rightarrow 3'$ WK automaton of type $\textbf{F}$ accepting the language $\{a^{2n+q} c^{4m}b^{2q+n}\mid n,q\geq 0,  m\in \{0,1\} \}$.}
    \label{tm:5}
\end{figure}

The following result also shows that the new model differs from the one
that is using the sensing parameter in its transitions.
\begin{theorem}\label{thm:6}
	The language class accepted by $\textnormal{\textbf{F}}$ sensing $5' \rightarrow 3'$ WK automata is properly included in the language class of sensing $5' \rightarrow 3'$ WK automata.
\end{theorem}

\begin{proof}
	The language $L=\{a^n cb^n c\mid n\geq 1\}$ can be accepted by a sensing $5' \rightarrow 3'$ WK automaton (without restrictions) (see Figure \ref{tm:6}).
Now we show that there is no $\textbf{F}$ sensing $5'\rightarrow3'$ WK automaton which accepts $L$. Assume the contrary that the language $L$ is accepted by an $\textbf{F}$ sensing $5' \rightarrow 3'$ WK automaton. Let the radius of the automaton be $r$.
Let $w=a^mcb^mc \in L$ with  $m\geq r$. 
 Thus the word $w$ cannot be accepted by applying exactly one transition from the initial state $q_0$.
Now, suppose that there exists $q\in \delta(q_0,w_1,w_2)$ such that $w$ can be accepted by using transition(s) from $q$. Since in $\textbf{F}$ sensing $5' \rightarrow 3'$ WK automaton all states are final, then the concatenation of $w_1$ and $w_2$ is accepted, thus, it must be in $L$ (i.e. $w_1w_2\in L$). Therefore $w_1w_2=a^{m'}cb^{m'}c$ where $2m'+2\leq r\leq m$. To expand both blocks $a^+$ and $b^+$ to continue the accepting path of $w$, the left head must be before/in/right after the subword $a^{m'}$, and the right head must be right before/in/right after the subword $b^{m'}$. However, this is contradicting the fact that the two heads together already read $a^{m'}cb^{m'}c$.
Hence, it is not possible to accept $w$ by an $\textbf{F}$ sensing $5' \rightarrow 3'$ WK automaton and the language $L$ cannot be accepted by an $\textbf{F}$ sensing $5' \rightarrow 3'$ WK automaton.
\end{proof}
\begin{figure}[h]
    \centering
        \includegraphics[scale=0.15]{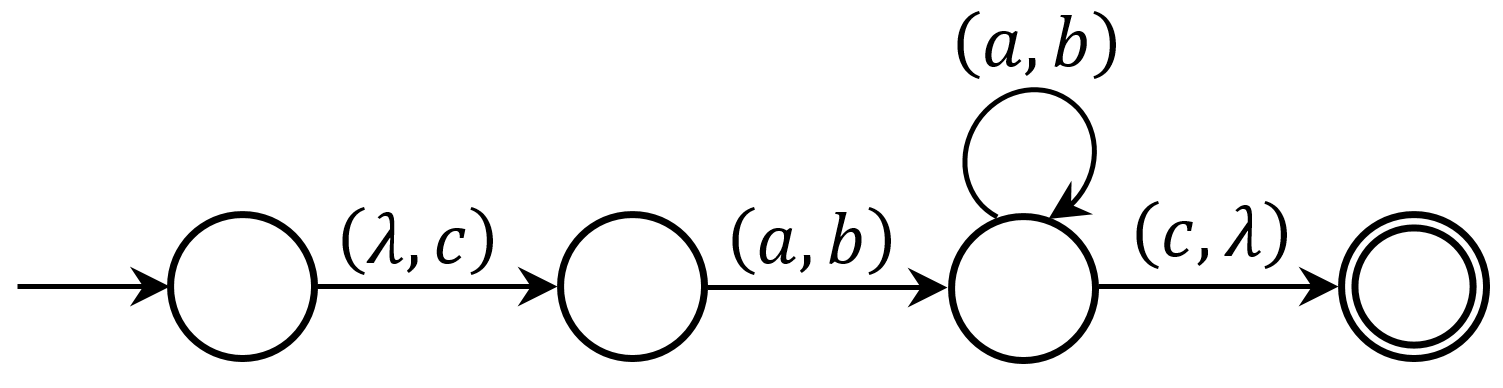}
            \caption{A sensing $5' \rightarrow 3'$ WK automaton accepts the language $\{a^n cb^n c\mid n\geq 1\}$.}
    \label{tm:6}
\end{figure}
\begin{proposition} \label{c3}
The language $L=\{a^n b^m\mid n=m$ or $n=m+1\}$ can be accepted by $\textnormal{\textbf{F1}}$ sensing  $5' \rightarrow 3'$ WK automata, but cannot be accepted by $\textnormal{\textbf{N1}}$, $\textnormal{\textbf{NS}}$ and $\textnormal{\textbf{N}}$ sensing $5' \rightarrow 3'$ WK automata.
\end{proposition}

\begin{proof}
As it is shown in Figure \ref{c:3}, $L$ can be accepted by an $\textbf{F1}$ sensing  $5' \rightarrow 3'$ WK automaton.
Suppose that $L$ can be accepted by an $\textbf{N}$ sensing  $5' \rightarrow 3'$ WK automaton. The $w_s$ of $L$ is $a$, therefore at least one of the loop-transitions $(a,\lambda)$ and $(\lambda,a)$ is possible from the only state. 
 Since this automaton has only one state, using any of these transitions leads to accept $a^n$ for any $n\geq2$ which are not in $L$. Thus this language cannot be accepted by an $\textbf{N}$, $\textbf{N1}$, $\textbf{NS}$ sensing $5' \rightarrow 3'$ WK automaton.
\end{proof}

\begin{figure}[h]
    \centering
        \includegraphics[scale=0.15]{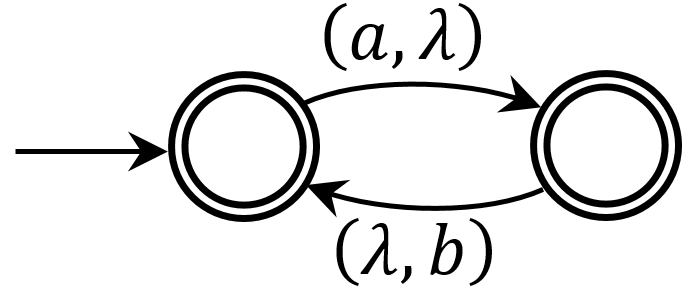}
            \caption{An $\textbf{F1}$ sensing $5' \rightarrow 3'$ WK automaton accepts the language $L=\{a^n b^m\mid n=m$ or $n=m+1\}$.}
    \label{c:3}
\end{figure}

\begin{remark}\label{R2}The following statements follow from Proposition \ref{c3}:
	\begin{enumerate}[label=(\alph*)]
		\item The class of languages that can be accepted by $\textnormal{\textbf{N1}}$ sensing $5'\rightarrow 3'$ WK automata is properly included in the language class accepted by $\textnormal{\textbf{F1}}$ sensing $5'\rightarrow 3'$ WK automata.
		\item The class of languages that can be accepted by $\textnormal{\textbf{NS}}$ sensing $5'\rightarrow 3'$ WK automata is properly included in the language class accepted by $\textnormal{\textbf{FS}}$ sensing $5'\rightarrow 3'$ WK automata.
		\item The class of languages that can be accepted by $\textnormal{\textbf{N}}$ sensing $5'\rightarrow 3'$ WK automata is properly included in the language class accepted by $\textnormal{\textbf{F}}$ sensing $5'\rightarrow 3'$ WK automata.
	\end{enumerate}
\end{remark}

\subsection{Incomparability results}
\begin{theorem}\label{thm:18}
The class of languages that can be accepted by $\textnormal{\textbf{N}}$ sensing $5' \rightarrow 3'$ WK automata is incomparable with the classes of languages that can be accepted by $\textnormal{\textbf{FS}}$ and $\textnormal{\textbf{F1}}$ sensing $5' \rightarrow 3'$ WK automata under set theoretic inclusion.
\end{theorem}

\begin{proof}
The language $L=\{ww^R\mid w\in\{a,b\}^*\}$ can be accepted by an $\textbf{N}$ sensing  $5' \rightarrow 3'$ WK automaton (Figure \ref{tm:11A}). Suppose that an $\textbf{FS}$ sensing  $5' \rightarrow 3'$ WK automaton accepts $L$. Let the radius of this automaton be $r$. Let $w_2=w_1w_1^R\in L$ with $w_1=(bbbaaa)^m$ and $m>r$. The word $w_2$ cannot be accepted by using only one of the transitions from the initial state $q_0$, i.e., $\delta(q_0,w_1w_1^R,\lambda)$ or $\delta(q_0,\lambda,w_1w_1^R)$ is not possible (because the length of $w_2$ ). Therefore there exists either $q\in\delta(q_0, w_3w_3^R,\lambda)$, $w_3\in V^*$ or $q\in\delta(q_0, \lambda,w_3w_3^R)$, $w_3\in V^*$ such that $w_2$ can be accepted by using transition(s) from $q$.
Since the word $w_3w_3^R$ should be in the language $L$ (i.e., it is an even palindrome) and the length of $bbb$ and $aaa$ patterns in $w_2$ is odd, the only even palindrome proper prefix (suffix) of $w_2$ is $bb$. Thus $w_3w_3^R=bb$ must hold.

Without loss of generality, assume that there exists $q\in\delta(q_0,bb, \lambda)$ in the automaton.
 By continuing the process,
 we must have at least one of $q'\in \delta(q,w_4,\lambda)$ or
 $q'\in \delta(q,\lambda,w_4)$ such that $bb w_4 \in L$ and $w_4$ is either the prefix or the suffix of the remaining unread part of word $w_2$, i.e., $ba^3(b^3a^3)^{m-1}(a^3b^3)^m$, with length less than $m$. Clearly, $w_4$ cannot be a prefix, and it can be only the suffix $bb$.
 Thus, in $q'$ the unprocessed part of the input is $ba^3(b^3a^3)^{m-1}(a^3b^3)^{m-1}a^3b$.
 Now the automaton must read a prefix or a suffix of this word, let us say $w_5$ such that $bbw_5bb \in L$, that is $w_5$ itself is an even palindrome, and its length is at most $r< m$.
 But such a word does not exist,
the length of $bbb$ and $aaa$ patterns in the unread part is odd and their length is more than $r$.
We have arrived to a contradiction, thus $L$ cannot be accepted by any
$\textbf{FS}$ sensing  $5' \rightarrow 3'$ WK automaton.

To prove the other direction, let us consider the language $L=\{a^n b^m\mid n=m \text{ or } n=m+1\}$. This language can be accepted by an $\textbf{F1}$ sensing $5' \rightarrow 3'$ WK automaton as it is shown in Figure \ref{c:3}. Moreover, by Proposition \ref{c3} this language cannot be accepted by any $\textbf{N}$ sensing $5' \rightarrow 3'$ WK automata.
\end{proof}

\begin{figure}[h]
    \centering
        \includegraphics[scale=0.15]{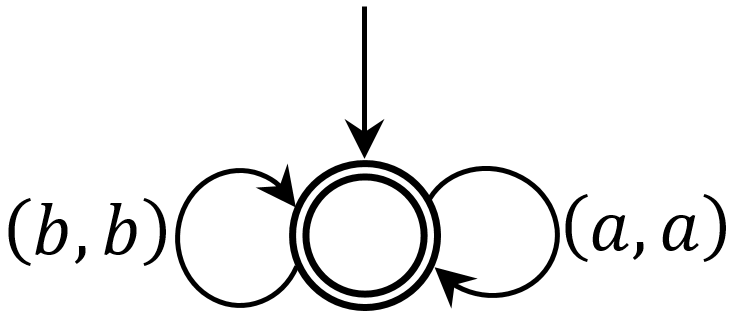}
            \caption{A sensing $5' \rightarrow 3'$ WK automaton of type $\textbf{N}$ accepting the language of even palindromes $\{ww^R\mid w\in\{a,b\}^* \}$.
}
    \label{tm:11A}
\end{figure}

\begin{theorem}\label{thm:19}
The language class accepted by $\textnormal{\textbf{NS}}$ sensing $5' \rightarrow 3'$ WK automata is incomparable with the language class accepted by $\textnormal{\textbf{F1}}$ sensing $5' \rightarrow 3'$ WK automata.
\end{theorem}

\begin{table}[t]
\centering
\captionof{table}{Some specific languages belonging to language classes accepted by various classes of WK automata. Reference to figures indicate a specific automaton that accept the given language. \xmark \ indicates that the language cannot be accepted by the automata type of the specific column. Trivial inclusions are also shown, e.g., in the first line $\textbf{N1}$ in, e.g., column $\textbf{F}$ means that every $\textbf{N1}$ automaton is, in fact, also an $\textbf{F}$ automaton.} \label{table:1}
\renewcommand{\arraystretch}{1.5}
\begin{tabular}{P{5.61cm} | c  c  c c c  c  c c  c  c  c  c  }
Language              & $\textbf{N1}$ & $\textbf{NS}$ & $\textbf{N}$ & $\textbf{F1}$ & $\textbf{FS}$  &$\textbf{F}$ & WK \\
\hline
$\{a^nb^m\mid n,m\geq 0 \}$          &Fig. \ref{tm:2a}& $\textbf{N1}$ &$\textbf{N1}$ & $\textbf{N1}$ & $\textbf{N1}$ & $\textbf{N1}$  &$\textbf{N1}$   \\
$\{a^{3n}b^{2m}\mid n,m\geq 0 \}$          &\xmark& Fig. \ref{tm:2} & $\textbf{NS}$ & \xmark & $\textbf{NS}$ & $\textbf{NS}$  &$\textbf{NS}$   \\
$\{ww^R\mid w\in\{a,b\}^* \}$          &\xmark& \xmark &Fig. \ref{tm:11A} & \xmark &\xmark & $\textbf{N}$  & $\textbf{N}$   \\
$\{a^nb^m\mid n=m$ or $n=m+1 \}$     &\xmark& \xmark &\xmark & Fig. \ref{c:3} & $\textbf{F1}$ & $\textbf{F1}$  &$\textbf{F1}$   \\
$\{(aa)^n(bb)^m\mid m\leq n \leq m+1,m\geq 0 \}$          &\xmark& \xmark &\xmark & \xmark& Fig. \ref {tm:4} & $\textbf{FS}$  & $\textbf{FS}$   \\
$\{a^{2n+q}c^{4m}b^{2q+n}\mid n,q\geq 0, m\in \{0,1\} \}$ 	 &\xmark& \xmark &\xmark & \xmark & \xmark & Fig. \ref{tm:5}  &$\textbf{F}$   \\
$\{a^n cb^n c\mid n\geq 1\}$           &\xmark& \xmark &\xmark & \xmark & \xmark & \xmark  &Fig. \ref{tm:6}   \\
\end{tabular} \bigskip

\end{table}

\begin{proof}
Consider the language $L=\{a^{3n}b^{2m}\mid n,m\geq0\}$. An $\textbf{NS}$ sensing $5' \rightarrow 3'$ WK automaton can move one of its heads at a time. Therefore it can read three $a$'s by the left head or two $b$'s by the right head (see Figure \ref{tm:2}). Although, according to Lemma \ref{l1}, $w_s$  is $bb$ and it cannot be the shortest nonempty accepted word for an $\textbf{F1}$ sensing $5' \rightarrow 3'$ WK automaton. Therefore, this language cannot be accepted by an $\textbf{F1}$ sensing $5' \rightarrow 3'$ WK automaton.

Now let us consider the language $L=\{a^n b^m\mid n=m \text{ or } n=m+1\}$. This language can be accepted by an $\textbf{F1}$ sensing $5' \rightarrow 3'$ WK automaton as it is shown in Figure \ref{c:3}. By Proposition \ref{c3}, it is already shown that $L$ cannot be accepted by any $\textbf{N}$ sensing $5' \rightarrow 3'$ WK automata and obviously it cannot be accepted by any $\textbf{NS}$ sensing $5' \rightarrow 3'$ WK automata, neither.
\end{proof}


\section{Conclusion}

\begin{figure}[t]
    \centering
        \includegraphics[scale=0.41]{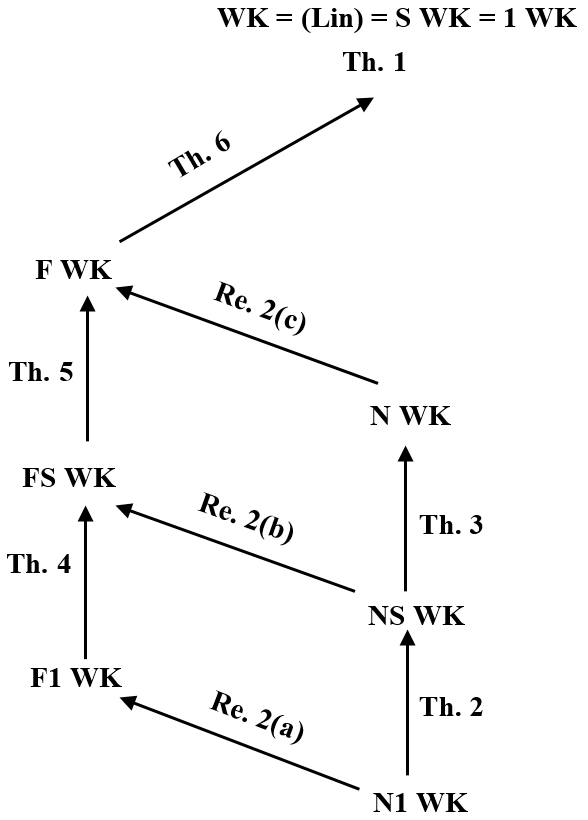}
            \caption{Hierarchy of sensing $5' \rightarrow 3'$ WK finite automata languages in a Hasse diagram (the language classes accepted by various types of sensing $5' \rightarrow 3'$ WK finite automata, the types of the automata are displayed in the figure with the abbreviations: $\textbf{N}$: stateless, $\textbf{F}$: all-final, $\textbf{S}$: simple, $\textbf{1}$: 1-limited; Lin stands for the class of linear context-free languages).
            Labels on the arrows indicate where the proof of the proper containment was presented (Th stands for Theorems, Re stands for Remark). The language classes for which the containment is not shown are incomparable.}

    \label{hasse}
\end{figure}

Comparing the new model to the old models we should mention that the general model (the automata without using any restrictions) has the same accepting power, i.e., the linear context-free languages, as the old sensing $5' \rightarrow 3'$ WK automata model with sensing parameter.  However, by our proofs, the new model gives a more finer hierarchy, as it is displayed in Figure \ref{hasse}. Table \ref{table:1} gives some specific languages that separate some of the language classes.
Further comparisons of related language classes and properties of the language classes defined by the new model are left to the future: in a forthcoming paper the deterministic variants are addressed.
It is also an interesting idea to see the connections of other formal models and our WK automata, e.g.,
similarities with some variants of Marcus contextual grammars \cite{MarcusCG} can be established.

\section*{Acknowledgements} The authors are very grateful to the anonymous reviewers for their comments and remarks.

\nocite{*}
\bibliographystyle{eptcs}
\bibliography{generic}

\begin{thebibliography}{10}
\providecommand{\bibitemdeclare}[2]{}
\providecommand{\surnamestart}{}
\providecommand{\surnameend}{}
\providecommand{\urlprefix}{Available at }
\providecommand{\url}[1]{\texttt{#1}}
\providecommand{\href}[2]{\texttt{#2}}
\providecommand{\urlalt}[2]{\href{#1}{#2}}
\providecommand{\doi}[1]{doi:\urlalt{http://dx.doi.org/#1}{#1}}
\providecommand{\bibinfo}[2]{#2}

\bibitemdeclare{article}{Adleman}
\bibitem{Adleman}
\bibinfo{author}{Leonard~M. \surnamestart Adleman\surnameend}
  (\bibinfo{year}{1994}): \emph{\bibinfo{title}{Molecular computation of
  solutions to combinatorial problems}}.
\newblock {\sl \bibinfo{journal}{Science}} \bibinfo{volume}{226}, pp.
  \bibinfo{pages}{1021--1024}, \doi{10.1126/science.7973651}.

\bibitemdeclare{article}{Czeizle}
\bibitem{Czeizle}
\bibinfo{author}{Elena \surnamestart Czeizler\surnameend} \&
  \bibinfo{author}{Eugen \surnamestart Czeizler\surnameend}
  (\bibinfo{year}{2006}): \emph{\bibinfo{title}{A Short Survey on Watson-Crick
  Automata}}.
\newblock {\sl \bibinfo{journal}{Bulletin of the {EATCS}}}
  \bibinfo{volume}{88}, pp. \bibinfo{pages}{104--119}.

\bibitemdeclare{inproceedings}{Freund}
\bibitem{Freund}
\bibinfo{author}{Rudolf \surnamestart Freund\surnameend},
  \bibinfo{author}{Gheorghe \surnamestart {P\u aun}\surnameend},
  \bibinfo{author}{Grzegorz \surnamestart Rozenberg\surnameend} \&
  \bibinfo{author}{Arto \surnamestart Salomaa\surnameend}
  (\bibinfo{year}{1997}): \emph{\bibinfo{title}{Watson-Crick finite automata}}.
\newblock In: {\sl \bibinfo{booktitle}{3rd DIMACS Sympozium On DNA Based
  Computers}}, \bibinfo{address}{Philadelphia}, pp. \bibinfo{pages}{305--317},
\doi{10.1090/dimacs/048/22}.

\bibitemdeclare{inproceedings}{Kuske}
\bibitem{Kuske}
\bibinfo{author}{Dietrich \surnamestart Kuske\surnameend} \&
  \bibinfo{author}{Peter \surnamestart Weigel\surnameend}
  (\bibinfo{year}{2004}): \emph{\bibinfo{title}{The role of the complementarity
  relation in Watson-Crick automata and sticker systems}}.
\newblock In: {\sl \bibinfo{booktitle}{Developments in Language Theory, DLT
  2004}}, {\sl \bibinfo{series}{Lecture Notes in Computer Science, LNCS}}
  \bibinfo{volume}{3340}, \bibinfo{publisher}{Springer, Berlin, Heidelberg},
  pp. \bibinfo{pages}{272--283}, \doi{10.1007/978-3-540-30550-7_23}.

\bibitemdeclare{article}{Leupold}
\bibitem{Leupold}
\bibinfo{author}{Peter \surnamestart Leupold\surnameend} \&
  \bibinfo{author}{Benedek \surnamestart Nagy\surnameend}
  (\bibinfo{year}{2010}): \emph{\bibinfo{title}{$5' \rightarrow 3'$
  Watson-Crick automata with several runs}}.
\newblock {\sl \bibinfo{journal}{Fundamenta Informaticae}}
  \bibinfo{volume}{104}, pp. \bibinfo{pages}{71--91},
\doi{10.3233/FI-2010-336}.

\bibitemdeclare{inproceedings}{DNA2008}
\bibitem{DNA2008}
\bibinfo{author}{Benedek \surnamestart Nagy\surnameend} (\bibinfo{year}{2008}):
  \emph{\bibinfo{title}{On $5' \rightarrow 3'$ sensing Watson-Crick finite
  automata}}.
\newblock In: {\sl \bibinfo{booktitle}{Garzon M.H., Yan H. (eds): DNA
  Computing. DNA 2007: Selected revised papers}}, {\sl \bibinfo{series}{Lecture
  Notes in Computer Science, LNCS}} \bibinfo{volume}{4848},
  \bibinfo{publisher}{Springer, Berlin, Heidelberg}, pp.
  \bibinfo{pages}{256--262}, \doi{10.1007/978-3-540-77962-9\_27}.

\bibitemdeclare{inproceedings}{Nagy2009}
\bibitem{Nagy2009}
\bibinfo{author}{Benedek \surnamestart Nagy\surnameend} (\bibinfo{year}{2009}):
  \emph{\bibinfo{title}{On a hierarchy of $5' \rightarrow 3'$ sensing WK finite
  automata languages}}.
\newblock In: {\sl \bibinfo{booktitle}{Computaility in Europe, CiE 2009:
  Mathematical Theory and Computational Practice, Abstract Booklet,
  Heidelberg}}, pp. \bibinfo{pages}{266--275}.

\bibitemdeclare{inbook}{iConcept}
\bibitem{iConcept}
\bibinfo{author}{Benedek \surnamestart Nagy\surnameend} (\bibinfo{year}{2010}):
  \emph{\bibinfo{title}{$5' \to 3'$ Sensing Watson-Crick Finite Automata}}, pp.
  \bibinfo{pages}{39--56}.
\newblock {\sl \bibinfo{series}{In: Gabriel Fung (ed.): Sequence and Genome
  Analysis II –- Methods and Applications}}~, \bibinfo{publisher}{iConcept
  Press}.

\bibitemdeclare{article}{Nagy2013}
\bibitem{Nagy2013}
\bibinfo{author}{Benedek \surnamestart Nagy\surnameend} (\bibinfo{year}{2013}):
  \emph{\bibinfo{title}{On a hierarchy of $5' \rightarrow 3'$ sensing
  Watson-Crick finite automata languages}}.
\newblock {\sl \bibinfo{journal}{Journal of Logic and Computation}}
  \bibinfo{volume}{23}(\bibinfo{number}{4}), pp. \bibinfo{pages}{855--872},
\doi{10.1093/logcom/exr049}.

\bibitemdeclare{book}{MarcusCG}
\bibitem{MarcusCG}
\bibinfo{author}{Gheorghe \surnamestart {P\u aun}\surnameend}
  (\bibinfo{year}{1997}): \emph{\bibinfo{title}{Marcus Contextual Grammars}}.
\newblock {\sl \bibinfo{series}{Studies in Linguistics and Philosophy}}
  \bibinfo{volume}{Volume 67}, \bibinfo{publisher}{Kluwer},
  \bibinfo{address}{Dordrecht}, \doi{10.1007/978-94-015-8969-7\_2}.

\bibitemdeclare{book}{Paun}
\bibitem{Paun}
\bibinfo{author}{Gheorghe \surnamestart {P\u aun}\surnameend},
  \bibinfo{author}{Grzegorz \surnamestart Rozenberg\surnameend} \&
  \bibinfo{author}{Arto \surnamestart Salomaa\surnameend}
  (\bibinfo{year}{2002}): \emph{\bibinfo{title}{DNA Computing: New Computing
  Paradigms}}.
\newblock \bibinfo{publisher}{Springer-Verlag},
\doi{10.1007/978-3-662-03563-4}.

\bibitemdeclare{book}{Handb}
\bibitem{Handb}
\bibinfo{editor}{Grzegorz \surnamestart Rozenberg\surnameend} \&
  \bibinfo{editor}{Arto \surnamestart Salomaa\surnameend}, editors
  (\bibinfo{year}{1997}): \emph{\bibinfo{title}{Handbook of Formal Languages}}.
\newblock \bibinfo{publisher}{Springer}.

\end{thebibliography}

\end{document}